\documentclass[10pt,twocolumn,UKenglish]{article}

\newcommand{\cD}{{\mathcal D}}

 \usepackage{amstext}
\usepackage{epsfig}
\usepackage{subfig}
\usepackage{psfrag}
\usepackage{comment}
\usepackage{amsmath}
\usepackage{fullpage}
\usepackage{color}
\usepackage{arcs}
\usepackage{algorithmic}
\usepackage{graphics}
\usepackage[ruled,vlined]{algorithm2e}
\usepackage{amssymb,amsmath}

\usepackage{graphicx}
\usepackage{caption}
\usepackage{amssymb}

\newtheorem{lemma}{Lemma}
\newtheorem{definition}{Definition}
\newtheorem{theorem}{Theorem}
\newtheorem{corollary}{Corollary}
\newtheorem{claim}{Claim}
\newtheorem{observation}{Observation}

\newenvironment{proof}{\textbf{Proof:}}{\hfill$\Box$}

\newcommand{\dan}{\emph{DAN}}

\newcommand{\kaushik}[1]{\textcolor{red}{(Kau: #1)}}

\newcommand{\stefan}[1]{\textcolor{blue}{(Stef: #1)}}

\renewcommand\vec[1]{\overrightarrow{#1}}
\newcommand\cev[1]{\overleftarrow{#1}}

\def\epl{\mathrm{EPL}}
\def\apd{\mathrm{APD}}
\def\nd{\mathrm{ND}}
\def\d{\mathrm{d}}
\def\p{\mathrm{p}}
\def\bE{\mathbb{E}}
\def\bnd{\mathrm{BND}}

\def\card#1{\lvert #1 \rvert}
\def\spanner{S}

\def\dan{{\tt DAN}}
\def\dang{N}
\def\danfam{\mathcal N}
\def\maxd{\Delta_{\max}}
\def\avgd{\Delta_{\mathrm{avg}}}

\title{Bounded-Degree Network Designs of Bounded Degree}

\title{Demand-Aware Network Designs of Bounded Degree\footnote{
Research supported by the German-Israeli project
\emph{Self-Adjusting Virtual Infrastructures} (G.I.F. No I-1245-407.6/2014).}}

\author{Chen Avin$^1$ \quad Kaushik Mondal$^1$ \quad 
Stefan Schmid$^2$\\
{\small $^1$ Ben Gurion University of the Negev, Israel \quad 
  Aalborg University, Denmark}}
  
  \date{}

\begin{document}

\maketitle


\begin{abstract}
Traditionally, networks such as datacenter interconnects
are designed to optimize worst-case performance
under \emph{arbitrary} traffic patterns.
Such network designs can however be far from optimal when
considering the \emph{actual}
workloads and traffic
patterns which they serve.
This insight led to the development of demand-aware datacenter interconnects which can be reconfigured depending on the workload.

Motivated by these trends, this paper initiates the algorithmic study of
demand-aware networks (\dan s) designs, and in particular the design of
bounded-degree networks.
The inputs to the network design problem are a discrete communication request distribution, $\cD$,  defined over communicating pairs from
the node set $V$, and a bound, $\Delta$, on the maximum degree.
In turn, our objective is to
design an (undirected) demand-aware network $\dang=(V,E)$
of bounded-degree~$\Delta$,
which provides
short routing paths between frequently communicating nodes
distributed across $\dang$.
In particular, the designed network should minimize the \emph{expected path length} on $\dang$
(with respect  to $\cD$), which is a basic measure of the efficiency of the network.

We show that this fundamental network design problem exhibits interesting connections
to several classic combinatorial problems and to information theory.
We derive a general lower bound based
on the entropy of the communication pattern $\cD$, and present asymptotically optimal
network-aware design algorithms for important distribution families, such as sparse
distributions and distributions of locally bounded
doubling dimensions.
\end{abstract}

\section{Introduction}\label{sec:intro}


The problem studied in this paper is motivated by
the advent of
more flexible datacenter interconnects, such as ProjectToR~\cite{Jia2017,ghobadi2016projector}.
These interconnects aim to overcome a fundamental drawback of
traditional datacenter network designs: the fact that network designers
must decide \emph{in advance} on how much capacity to provision
between electrical packet switches, e.g., between
top-of-rack (ToR) switches in datacenters. This leads to an undesirable
tradeoff~\cite{fat-free}: either capacity is over-provisioned and therefore the interconnect
expensive (e.g., a fat-tree provides full-bisection bandwidth),
or one may risk congestion, resulting in a poor cloud application performance.
Accordingly, systems such as ProjectToR provide a reconfigurable interconnect,
allowing to establish links flexibly and
in a \emph{demand-aware manner}. For example,
direct links or at least short communication paths can be established between
frequently communicating ToR switches.
Such links can be implemented using
a bounded number of lasers, mirrors, and photodetectors
per node~\cite{ghobadi2016projector}.
First experiments with this technology demonstrated promising
results: although the interconnecting networks is of
bounded degree, short routing paths
can be provided between communicating nodes.

The problem of designing demand-aware
networks is a fundamental one, and finds interesting
 applications in many distributed and networked systems.
 For example, while many peer-to-peer overlay networks today are designed towards optimizing the \emph{worst-case performance}
(e.g., minimal diameter and/or degree), it is an intriguing question
whether the ``hard instances'' actually show up in real life,
and whether better topologies can be designed if we are given more information about the actual communication patterns 
these networks serve in practice.

While the problem is natural,
surprisingly little is known today about the design of demand-aware networks.
At the same time, as we will show in this
paper, the design of demand-aware networks is related to
several classic combinatorial problems.

Our vision is reminiscent in spirit to the question posed by Sleator and Tarjan
over 30 years ago in the context of binary search trees~\cite{Demaine:2009:GBS:1496770.1496825,Sleator85}:
While there is an inherent lower bound of $\Omega(\log{n})$ for accessing
an arbitrary element in a binary search tree, can we do better on some ``easier'' instances? The authors identified the \emph{entropy} to be a natural metric to measure the performance under actual demand patterns.
We will provide evidence in this paper that the entropy, in a slightly different flavor, also plays a crucial role in the context of network designs, establishing an interesting connection.

\noindent \textbf{The Problem: Bounded Network Design.}
We consider the following network design problem, henceforth referred to as the \emph{Bounded Network Design} problem, short  \emph{BND}.
We consider a set of $n$ nodes (e.g., top-of-rack switches, servers, peers)
$V=\{1,\ldots,n\}$
interacting according to a certain \emph{communication pattern}.
The pattern is modelled by $\cD$, a discrete distribution over \emph{communication
requests} defined over $V \times V$.
We represent this distribution using a communication
matrix $M_{\cD}[\p(i,j)]_{n\times n}$ where
the $(i,j)$ entry indicates the communication frequency, $p(i,j)$, from the
(communication) source $i$ to the (communication) destination $j$. The matrix is normalized,
i.e., $\sum_{ij} \p(i,j)=1$.
Moreover, we can interpret the distribution $\cD$
as a weighted directed \emph{demand graph} $G_{\cD}$, defined over the same set of nodes $V$:
A directed edge $(u, v) \in E(G_{\cD})$ exists iff $\p(u, v) > 0$. We set the edge weight
to the communication probability: $w(i,j) = \p(i,j)$.

In turn, our objective is to design an unweighted,
undirected \emph{Demand-Aware Network}~(\dan)
defined over the set of nodes $V$ and the distribution $\cD$,
henceforth denoted as $\dang(\cD)$ or just $\dang$ when $\cD$ is clear from the context.
The objective is that $\dang(\cD)$ optimally serves
the communication requests from $\cD$ under the constraint that
$\dang$ must be chosen from a certain family of
\emph{desired topologies} $\danfam$.
In particular, we are interested in \emph{sparse} networks
(i.e., having a \emph{linear number} of edges)
with \emph{bounded} degree $\Delta$
(i.e., nodes have a small number of lasers~\cite{ghobadi2016projector}),
and we denote the family of $\Delta$-bounded degree graphs by $\danfam_{\Delta}$.

Note that the designed network can be seen as ``hosting''
the served communication pattern, i.e., the demand graph is embedded
on the designed network. Accordingly, we will sometimes refer to the
demand graph as the \emph{guest network} and to the designed
network as the  \emph{host network}.

Our objective is to minimize the \emph{expected path length}~\cite{Avin:2013:SGN:2689153.2689159,obst,Schmid16} of the designed
host network $\dang\in \danfam$:
For $u,v \in V(\dang)$, let $\d_ {\dang}(u,v)$ denote the shortest path between $u$ and $v$ in $\dang$.
Given a distribution $\cD$ over $V \times V$ and a graph $\dang$ over $V$, the \emph{Expected Path Length (EPL)} of route requests is defined as:
$$
\epl(\cD, \dang) = \bE_{\cD}[\d_ {\dang}(\cdot, \cdot)]=\sum_{(u,v) \in \cD} \p(u, v) \cdot \d_{\dang}(u,v)
$$

\begin{figure*}
\begin{center}
	\begin{tabular}{ccc}
	   \includegraphics[width=.3\textwidth]{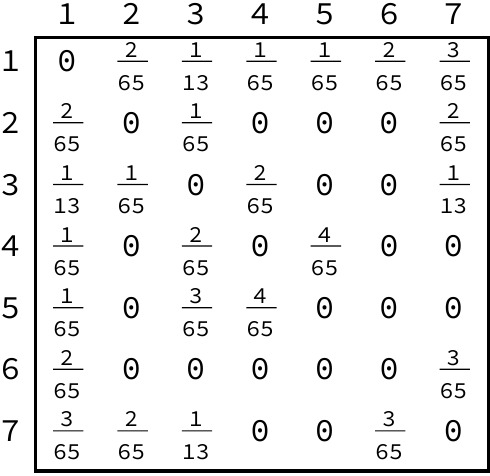} &
	   \includegraphics[width=.3\textwidth]{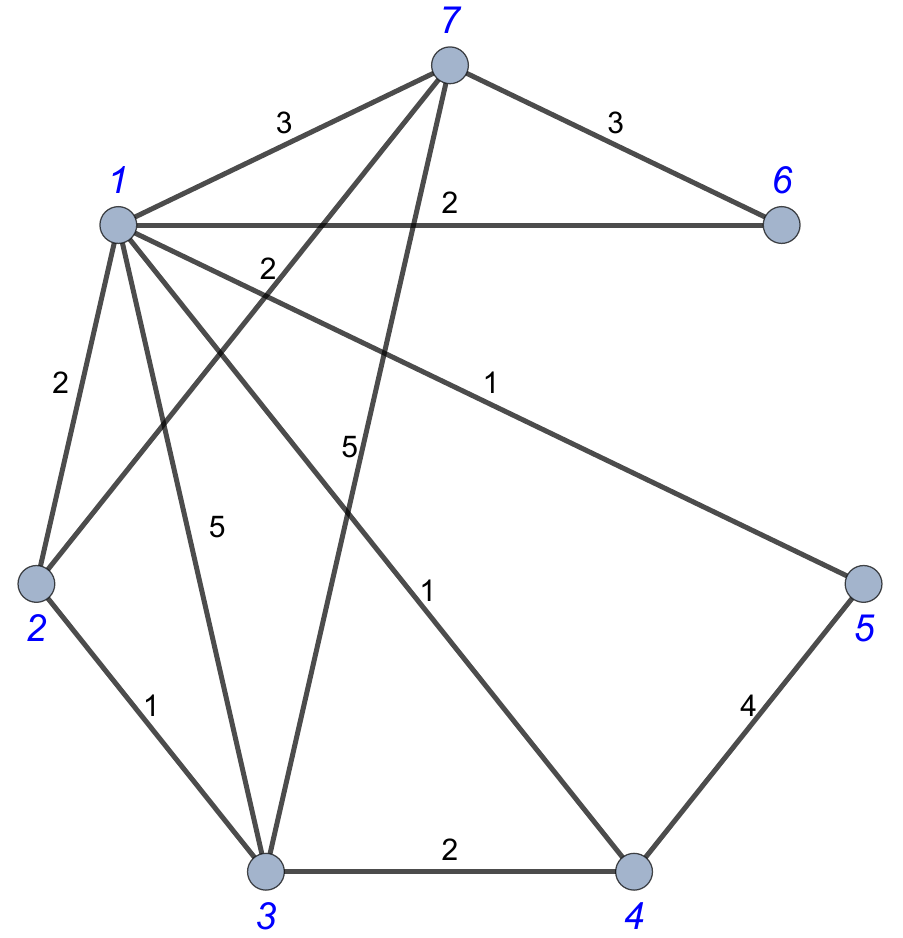}  &
	    \includegraphics[width=.3\textwidth]{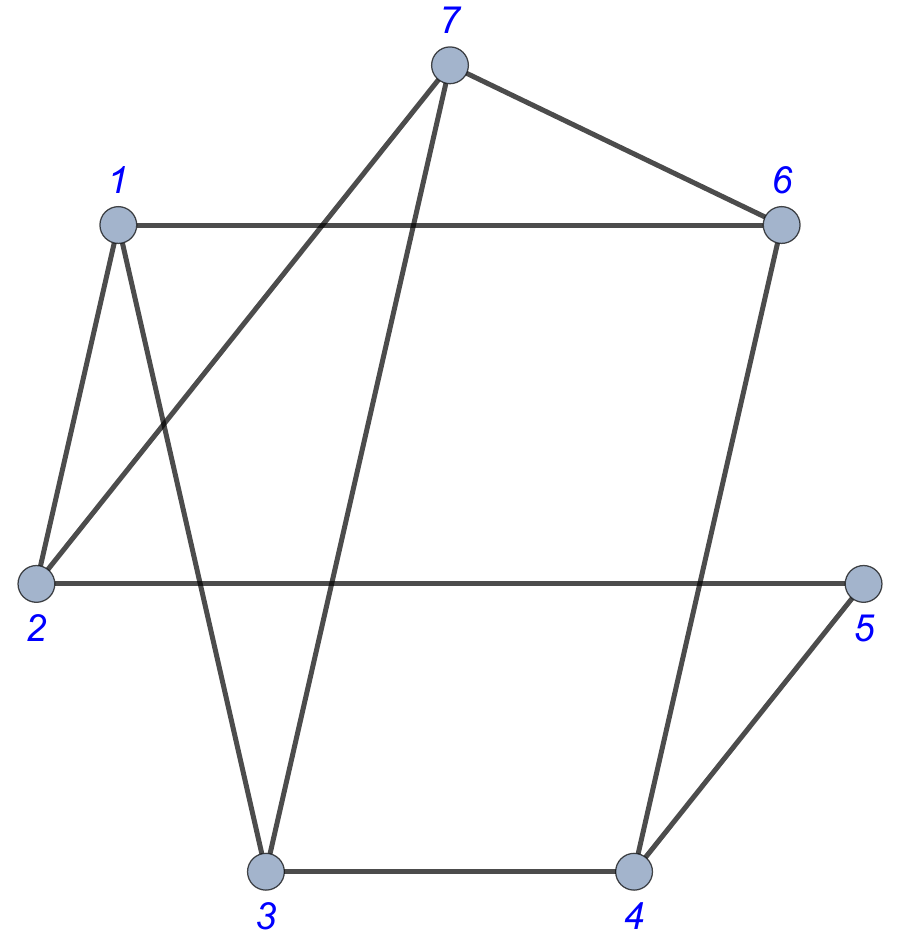} \\
	    (a) & (b) & (c)
		\end{tabular}
	\caption{Example of the \emph{bounded network design} problem. (a) A given demand distribution $\cD$ (which in this case is \emph{symmetric}). (b) The demand graph $G_{\cD}$ (with non-normalized weights).
	Nodes 1, 3, and 7 have a degree more than 3. (c) An optimal solution $\dan$~$\dang$ with $\Delta = 3$. In this case, the solution is not a subgraph but contains auxiliary edges
	(e.g., $\{2,5\}$), and $\epl(\cD, \dang) = 1.19$ while $H(X \mid Y) = 1.08$ (the
	Shannon entropy to the base 3 is $H(X)=1.68$).}
	\label{fig:model}
\end{center}
\end{figure*}

Since routing across the host network usually occurs along shortest paths,
the expected path length captures
the average hop-count of a route (e.g., latency incurred or energy consumed along the way).

Succinctly, the Bounded Network Design (BND) problem is to minimize the expected path length and is defined as follows:
\begin{definition}[\bf Bounded Network Design]
Given a communication distribution, $\cD$ and a maximum degree $\Delta$,
find a host graph $\dang \in \danfam_\Delta $ that minimizes the expected path length:
\begin{align*}
\bnd(\cD,\Delta) = \min_{\dang \in \danfam_\Delta} \epl(\cD, \dang)
\end{align*}
\end{definition}

\subsection{Our Contributions}
This paper initiates the study of a fundamental problem:
the design of demand-aware communication networks.
While our work is motivated by recent trends in datacenter network designs,
our model is natural and finds applications in many distributed and networked systems (e.g., peer-to-peer overlays).
The main contribution of this paper is to establish an interesting connection
of the network design problem to the conditional entropy of the communication matrix. In particular,
we present a lower bound
on the expected path length
of a network with maximum degree $\Delta$
which is proportional to the conditional
entropy of $\cD$,  $H_{\Delta}(X \mid Y) + H_{\Delta}(Y \mid X)$ where $\Delta$ is the base of the logarithm used for calculating the entropy.
While this lower bound can be as high as $\log n$,
for many distributions it can be much lower (even constant).
Our main results are presented in
Theorem~\ref{lem:sparsetransform} which proves
a matching upper bound for the case when $\cD$ is a sparse distribution, and
Theorem~\ref{thm:ldd} which proves a matching upper bound for the case when
$\cD$ is a regular and uniform (but maybe dense) distribution of a locally bounded doubling dimension.
Also in these two cases the conditional entropy could range from a constant and up to $\log n$.
At the heart of our technical contribution is a novel technique to
transform a low-distortion network of some high maximum degree 
to a low-degree
network whose maximum degree equals the average degree of the original network,
while maintaining an expected path length in the order of
the conditional entropy. Moreover, we
show an interesting reduction of uniform and regular distributions
to graph spanners in Theorem~\ref{th:spanner2design}.

%
%
%

\subsection{Paper Organization}

The remainder of this paper is organized as follows.
We put our work into perspective with respect to related work
in Section~\ref{sec:relwork} and introduce some preliminaries in Section~\ref{sec:preliminaries}.
We derive lower bounds in Section~\ref{sec:lowerbound} and
present algorithms to design networks for sparse distributions
resp.~regular and uniform distributions in Section~\ref{sec:sparse} resp.~Section~\ref{sec:reguni}.
We conclude our work and outline
directions for future research
in Section~\ref{sec:conclusion}.

\section{Putting Things Into Perspective and Related Work}\label{sec:relwork}

There are at least three interesting perspectives on our problem.
The first one arises when trying to gain some intuition about the problem complexity.
If $\Delta=n$, the problem is simple:
the demand (or guest) graph $G_{\cD}$
itself can be used as the host graph or $\dan$ $\dang \in \danfam_\Delta$,
providing an ideal expected path length 1.
If a sparse host graph is desired, a star topology
could be used as a $\dan$ to provide an expected path length of at most~2.
At the other end of the spectrum, if $\Delta=2$ (and the host network is required to be connected)
the $\dan$~$\dang$ must be a line  or a ring graph.
However, the problem of how to  arrange nodes
on the linear chain or the ring such that the expected path length is minimized,
is already NP-hard: the problem is essentially a
Minimum Linear Arrangment (MinLA) problem~\cite{charikar20062,feige2007improved,harper1964optimal}.
One perspective to see our contribution is that in this paper, we are
interested in what happens
between these extremes, for other values of $\Delta$, in particular for a constant $\Delta$ which guarantees that our host network
will be sparse, i.e., has a linear number of edges.
In contrast to the general arrangement problem which asks
for an embedding of the guest graph on a \emph{specific}
and \emph{given} host graph,
in our network design problem we are free to \emph{choose the best} host graph from a given family of graphs (i.e., bounded degree graphs).
One might wonder: does this flexibility make the problem
easier?

Sparse and distance-preserving spanners open a
second perspective on our work: intuitively, a good host
graph $\dang$ for $G_{\cD}$ ``looks similar'' to $G_{\cD}$.
But in contrast to classic spanner problems
in the literature which are primarily concerned with minimizing the worst-case \emph{distortion} (resp.~the average distortion)
among \emph{all} node pairs \cite{Chan2006,Peleg1989},
we are only interested in the \emph{local distortion}.
Namely, we aim to find a good ``spanner''
which preserves
\emph{locality of neighborhoods}, i.e., 1-hop
neighborhoods in the demand graph.
Second, unlike classic spanner problems but similar
to geometric (metric) spanners,
the designed network $\dang$ does not have to be a subgraph and may include edges which do not
exist in the demand network $G_{\cD}$, i.e., 0-entries in the corresponding communication matrix $M_{\cD}$.
We refer the corresponding edges as
\emph{auxiliary edges} (a.k.a.~shortcut edges~\cite{Meyerson2009}).
It is easy to see that auxiliary edges
can indeed be required to compute optimal network designs, and yield
strictly lower communication costs than subgraph spanners. Third,
in contrast to the frequently studied
sparse graph spanner problem variants,
we require that nodes in the designed network
are of degree at most $\Delta$.
Finally, we are not aware of any work studying the relationship between
spanners and entropy.
This makes our model fundamentally different from existing models
studied in the literature. 

The fact that our matrix represents a distribution
provides some interesting structure.
In particular, it leads us to a third connection, namely to
information and coding theory, see also~\cite{cacd} for a code-based
network deisgn of arbitrary degree.
It is known that the expected path length in binary search trees~\cite{Sleator85}
as well as in network designs providing local routing~\cite{obst,cacd,Schmid16} is related to the entropy
$H(X)$ (over the elements $X$ in the data structure) resp.~conditional entropy of the distribution:
$H(X|Y)+H(Y|X)$ is a lower bound on the expected path length of
local routing tree designs~\cite{Schmid16}
where $X$, $Y$ are the random variables distributed according to the marginal distribution of the sources and destinations in $\cD$.
This bound is tight for the limited case where $D$ is a product distribution (i.e., $p(i,j) = p(i)p(j)$).
Additionally the optimal binary search tree can be computed efficiently for every $\cD$ using dynamic programming~\cite{Schmid16}.
In the current work we extend this line of research by studying more general distributions  
and a larger family of host networks.

\section{Preliminaries}\label{sec:preliminaries}

We start with some notation about $\cD$. Let $\cD[i,j]$ or $p(i,j)$ denote the probability that source $i$ routes to destination $j$.
Let $p(i)$ denote the probability that $i$ is a source, i.e., $p(i)=\sum_j p(i,j)$. Similarly let $q(j)$ denote the probability that $j$ is a destination.
Let $X, Y$ be random variables describing the marginal distribution of the sources and destinations in $\cD$, respectively.
Let $\vec{\cD}[i]$ denote the normalized $i$'th row of $\cD$, that is, the probability distribution of destinations given that the source is $i$.
Similarly let $\cev{\cD}[j]$ denote the normalized $j$'th column of $\cD$, that is the probability distribution of sources given that the destination is $j$.
Let $Y_i$ and $X_j$ be random variables that are distributed according to $\vec{\cD}[i]$ and  $\cev{\cD}[j]$, respectively.
We say that $\cD$ is \emph{regular} if $G_{\cD}$ is a regular graph (both in terms of in and out degrees). We say that $\cD$ is \emph{uniform} if for every $\cD[i,j]>0$, $\cD[i,j]=\frac{1}{m}$ and $m$ is the number of edges in $G_{\cD}$. We say that $\cD$ is \emph{symmetric} if $\cD[i,j]=\cD[j,i]$.


We will show that a natural measure to assess the quality of a
designed network relates to the \emph{entropy} of the communication pattern.
For a discrete random variable $X$ with possible values
$\{x_1, \dots , x_n\}$, the entropy $H(X)$ of $X$ is defined as
\begin{align}
H(X) = \sum_{i=1}^n p(x_i)\log_2\frac{1}{p(x_i)}
\end{align}
where $p(x_i)$ is the
probability that $X$ takes the value $x_i$. Note that, $0 \cdot
\log_2\frac{1}{0}$ is considered as 0.
If $\bar{p}$ is a discrete distribution vector (i..e, $p_i \ge 0$ and $\sum_i p_i =1$) then we may write $H(\bar{p})$ or $H(p_1, p_2, \dots p_n)$ to denote the entropy of a random variable that is distributed according to $\bar{p}$.
If $\bar{p}$ is the uniform distribution with support $s$
($s$ being the number of places in the distribution with
$p_i > 0$, i.e., $p_i=1/s$) then  $H(\bar{p})= \log s$.

Using the decomposition (a.k.a.~grouping) properties of entropy the following is well-known~\cite{cover2012elements}:
\begin{align}\label{decompos1}
H(p_1, p_2, p_3 \dots p_m) \ge H(p_1 + p_2, p_3 \dots p_m)
\end{align}
and
\begin{align}\label{decompos2}
H(p_1, \dots,  p_m) & \ge (1-p_1)H(\frac{p_2}{1-p_1},\frac{p_3}{1-p_1} \dots \frac{p_m}{1-p_1})
\end{align}
For a joint distribution over $X, Y$, the \emph{joint entropy} is
defined as
\begin{align}
H(X, Y) = \sum_{i,j} p(x_i, y_j)\log_2\frac{1}{p(x_i,y_j)}
\end{align}
Also recall the definition of the
\emph{conditional entropy} $H(X \vert Y)$:
\begin{align}
H(X \vert Y) &= \sum_{i,j} p(x_i, y_j)\log_2\frac{1}{p(x_i \mid y_j)} 
\\ &= \sum_{j} p(y_j) \sum_{i} p(x_i \mid y_j )\log_2\frac{1}{p(x_i \mid y_j )}  \notag \\
&= \sum_{j=1}^n p(y_j)H(X \vert Y = y_j)
\end{align}
For $X$ and $Y$ defines as above from $\cD$ we also have
 \begin{align}
H(X \vert Y) &= \sum_{j=1}^n p(y_j)H(X \vert Y = y_j)  \\ &=  \sum_{j=1}^n q(j)H(\cev{\cD}[j]) \\ &=  \sum_{j=1}^n q(j)H(X_j)
\end{align}
 $H(Y \vert X)$ is defined similarly and we note that it may be the case that  $H(X \vert Y) \neq H(Y \vert X)$.
We may simply write $H$ for the entropy if
the purpose is given by the context.
By default, we will denote by $H$ the entropy computed using the binary
logarithm; if a different logarithmic basis $\Delta$ is used to compute the entropy,
we will explicitly write $H_\Delta$.

We recall the definition of a graph \emph{spanner}. Given a graph $G=(V,E)$, a subgraph $G'=(V,E')$ is a $t$-spanner of $G$ if for every
$u,v \in V$, $t \cdot d_G(u,v) \geq d_{G'}(u,v)$ and $t$ is the \emph{distortion} of the spanner.
We say that $G'=(V,E')$ is a \emph{graph metric} $t$-spanner if it is not a subgraph of $G$, i.e., it may have additional edges that are not in $G$.


\section{A Lower Bound}\label{sec:lowerbound}

We now establish an interesting connection to information theory
and show that the conditional entropy serves as a natural metric for bounded network
designs.
In particular, we prove that the expected path length BND$(\cD,\Delta)$ in any
demand-aware bounded network design, is at least in the order
of the conditional entropy. Formally:

\begin{theorem}
\label{th:lowerboundBND}
Consider the joint frequency distributions $\cD$.
Let $X$, $Y$ be the random variables distributed according to the marginal distribution of the sources and destinations in $\cD$, respectively.
Then
$$
\bnd(\cD,\Delta) \ge \Omega (\max(H_\Delta(Y|X), H_\Delta(X|Y))
$$
\end{theorem}



Before delving into the proof, let $\epl(\bar{p},T)$ denote the expected path length in a tree $T$ from the root to its nodes where the expectation it taking
over a distribution $\bar{p}$. That is $\epl(\bar{p},T) = \sum_i p_i d_T(root, i)$.
We recall the following well-known theorem:
\begin{theorem}[\cite{Mehlhorn75}, restated.]
\label{th:sseplH}
Let $H(\bar{p})$ be the entropy of the frequency distribution $\bar{p} = (p_1, p_2, \ldots ,p_n)$.
Let $T$ be an optimal binary search tree built over the above frequency distribution.
Then $\epl(\bar{p},T) \ge \frac{1}{\log 3}H(\bar{p})$.
\end{theorem}
Namely, the entropy $H(\bar{p})$, is a lower bound on the expected path length from the root to the nodes in the tree.
For higher degree graphs, we can extend the result:
\begin{lemma}
\label{lem:DeltaregularH}
Let $H_\Delta(\bar{p})$ be the entropy (calculated using the logarithm of base $\Delta$) of frequency distribution $\bar{p} = (p_1, p_2, \ldots ,p_n)$.
Let $T$ be an optimal $\Delta$-ary search tree built over the above frequency distribution.
Then, $\epl(\bar{p},T) \ge \frac{1}{\log(\Delta+1)}H_\Delta(\bar{p})$.
\end{lemma}
The proof almost directly follows from the proof of Theorem~\ref{th:sseplH} in~\cite{Mehlhorn75}, by extending  properties of binary trees to
$\Delta$-ary trees, see the appendix for details.
%
%
We now prove the lower bound.

\begin{proof}[Proof of Theorem \ref{th:lowerboundBND}]
The proof idea is to consider a network which is the union 
of the optimal
trees, one for each individual node. While the resulting network may
violate the degree requirement, it constitutes a valid lower bound.
So let us start by finding an optimal structure for each source
node $i$, according to its communications to all destination nodes from it, $\vec{\cD}[i]$.
Towards this end, we construct $n$ $\Delta$-ary  trees,
and let $T_\Delta^i$ be the optimal tree for source
node $i$ build using $\vec{\cD}[i]$.
From Lemma~\ref{lem:DeltaregularH}, we have:
\begin{eqnarray*}
\epl(\vec{\cD}[i],T_\Delta^i)&=\sum_{j=1}^{n} p(j|i)d_{T_\Delta^i}(i,j)
\\ &=\Omega (H_\Delta(Y \mid X=i))
\end{eqnarray*}

\noindent where $\epl(\vec{\cD}[i],T_\Delta^i)$ denotes the
expected path length of $T_\Delta^i$ given $\vec{\cD}[i]$ and $d_{T_\Delta^i}$ denotes the shortest path in $T_\Delta^i$. 
Now consider any bounded degree network $\dang_\Delta$ and compare it to
the forest $T$ made up of $n$ trees $T_\Delta^1, T_\Delta^2,\ldots,T_\Delta^n$. Then,

\begin{align*}
\epl(\cD,\dang_\Delta) &= \sum_{i=1}^{n} p(i) \cdot \epl(\vec{\cD}[i],\dang_\Delta)\\
&  \ge  \sum_{i=1}^{n} p(i) \cdot \epl(\vec{\cD}[i],T_\Delta^i) \\
&\ge\sum_{i=1}^{n} p(i) \cdot H_\Delta(Y \mid X=i) \\ & =\Omega (H_\Delta(Y|X))
\end{align*}
Similarly, we can consider a set of trees optimized toward the incoming communication of node $j$, $\cev{\cD}[j]$, and the marginal destination probability. We show:
$$
\epl(\cD,\dang_\Delta) \ge \Omega (H_\Delta(X \mid Y))
$$
Hence the theorem follows.
\end{proof}


\section{Network Design for Sparse Distributions}\label{sec:sparse}

We now present families of distributions which enable
\dan s that match the lower bound. Our approach will be based on replacing neighborhoods with near optimal binary (or $\Delta$-ary) trees.
Following the lower bound of Lemma \ref{lem:DeltaregularH}, it is easy to show a matching upper bound (for a constant $\Delta$).

\begin{lemma}\label{epl_tree}
Let $\bar{p}$ be a probability distribution on a set of node destinations (sources)  $V$, and let $u$ be a single source (destination) node.
Then one can design a tree $T$ with $u$ as a root node with degree one, connected to a $\Delta$-ary tree over $V$ such that the expected
path length from $u$ to all destinations (or from all sources to $u$) is:
\begin{align}
\epl(\bar{p},T) = \sum_i p_i \cdot d_T(u,i) \le O(H_{\Delta}(\bar{p}))
\end{align}
\end{lemma}

\begin{proof}
The proof follows by designing a Huffman $\Delta$-ary code over $\bar{p}$ (with expected code length less than $H_{\Delta}(\bar{p}) + 1$~\cite{cover2012elements}) and using it to build a rooted $\Delta$-ary tree.
While the nodes in the Huffman code are tree leaves,
we can move them up to become internal nodes, which only improves the
expected path length.
\end{proof}

\subsection{Tree Distributions}

A most fundamental class of distributions for which we can construct optimal network designs is based on trees.

\begin{theorem}
\label{lem:treetransform}
Let $\cD$ be a communication request distribution such that $G_{\cD}$ is a tree (i.e., ignoring the edge direction, $G_{\cD}$ forms a tree).
Let $X$, $Y$ be the random variables of the sources and destinations in $\cD$, respectively.
Then, it is possible to generate a  \dan~$\dang$ with maximum degree $8$,
such that
\begin{align*}
\epl(\cD,\dang) \le O(H(Y \mid X) + H(X \mid Y))
\end{align*}
This is asymptotically optimal.
\end{theorem}
\begin{proof}
We generate $\dang$ as follows. Consider an arbitrary node as the root of the tree $G_{\cD}$, call this tree $T_{\cD}$,
and consider the parent-child relationship implied by the root.
Let $\pi(i)$ denote the parent of node $i$.
Let $\vec{c}_i$ denote the communication distribution \emph{from} $v_i$ to its children (not including its single parent) and $\vec{\cD}[i]$ denote the communication distribution \emph{from} $i$ to its neighbors (children and parent). Let $p^{\pi}_i =\vec{\cD}[i][\pi(i)]$ denote the corresponding entry in $\vec{\cD}[i]$ for the parent of $i$.
From entropy Eq. \eqref{decompos2}, we have that $(1-p^{\pi}_i)H(\vec{c}_i) \le H(\vec{\cD}[i])$.
Similarly we define $\cev{c}_i$ and $\cev{\cD}[i]$ as the communication distribution \emph{to} $v_i$, from its children and neighbors respectively.

The construction has two phases. In the first phase we replace outgoing edges.
For each node $i$ replace the edges between $i$ and its \emph{children} with a binary tree according to $\vec{c}_i$ and the method of \cite{Mehlhorn75} (or Lemma \ref{epl_tree} for a general $\Delta$) for creating a near optimal binary tree. Let $\vec{B}_i$ denote this tree and recall that $\epl(\vec{c}_i, \vec{B}_i) \leq O(H(\vec{c}_i))$.
Note that every node $i$ may appear in at most two trees $\vec{B}_i$ and $\vec{B}_{\pi(i)}$; 
in $\vec{B}_i$ its degree is one and in $\vec{B}_{\pi(i)}$ its degree is at most 3, 
so the outgoing degree of each node is at most 4 after this phase.

In the second phase we take care of the remaining incoming edges from children to parents.
For each node $i$ replace the edges from its \emph{children} to it with a binary tree according to $\cev{c}_i$ and the method of \cite{Mehlhorn75} for creating a near optimal binary tree. Let $\cev{B}_i$ denote this tree and recall that $\epl(\cev{c}_i, \cev{B}_i) \leq O(H(\cev{c}_i))$.
Note that every node $i$ may appear in at most two trees $\cev{B}_i$ and  $\cev{B}_{\pi(i)}$; in $\cev{B}_i$ $i$'s degree is one and in $\cev{B}_{\pi(i)}$
$i$'s degree is at most 3. Thus, the incoming degree of each node is bounded by 4 after this phase.

Now we bound $\epl(\cD,\dang)$ by bounding the expected path lengths in the corresponding binary trees of each node,
first considering all edges from parent to children and then all edges from children to parent. Let $p(i)$ and $q(i)$ denote the probabilities that node $i$ will be
 a source and a destination of a communication request, respectively. Then:
\begin{align*}
\epl(\cD,\dang) &\leq \sum_{(u,v) \in \cD} p(u,v) d_{\dang}(u,v) \\
&=\sum_{(\pi(i),i) \in T_{\cD}} p(\pi(i),i) d_{\dang}(\pi(i),i) \\ & + \sum_{(i, \pi(i)) \in T_{\cD}} p(i, \pi(i)) d_{\dang}(i, \pi(i)) \\
&= \sum_{i=1}^n p(i) \epl(\vec{c}_i, \vec{B}_i) \\ & + \sum_{i=1}^n q(i) \epl(\cev{c}_i, \cev{B}_i) \\
&\leq \sum_{i=1}^n p(i) H(\vec{\cD}[i]) \\ & + \sum_{i=1}^n q(i) H(\cev{\cD}[i]) \\ & = H(Y \mid X) + H(X \mid Y)
\end{align*}

This matches our lower bound in Theorem~\ref{th:lowerboundBND}.
\end{proof}

\subsection{General Sparse Distributions}\label{sec:trafo}

Asymptotically optimal demand-aware networks can even be designed for general sparse distributions.

\begin{theorem}
\label{lem:sparsetransform}
Let $\cD$ be a communication request distribution where  $\avgd$  is the average degree in  $G_{\cD}$  (so the number of edges $m=\frac{\avgd \cdot n}{2}$).
Let $X$, $Y$ be the random variables of the sources and destinations in $\cD$, respectively.
Then, it is possible to generate a  \dan~$\dang$ with maximum degree $12\avgd$,
such that
\begin{align}
\epl(\cD,\dang) \le O(H(Y \mid X) + H(X \mid Y))
\end{align}
This is asymptotically optimal when $\avgd$ is a constant.
\end{theorem}
\begin{proof}
Recall that $G_{\cD}$ (for short $G$) is a directed graph and define in-degree and out-degree in the canonical way.
Let the (undirected) degree of a node, be the sum of its in-degree and out-degree and denote the average degree as $\avgd$.
Denote the $n/2$ nodes with the lowest degree in $G$ as \emph{low degree} nodes and the rest as \emph{high degree} nodes.
Note that each low degree node has a degree at most $2\avgd$ and both its in-degree and out-degree must be low.
A node with out-degree (in-degree) larger than $2\avgd$ is called a \emph{high out-degree} (\emph{high in-degree}) node (some nodes are neither low or high degree nodes).

The construction of $\dang$ will be done in two phases. In the first phase, we consider only (directed) edges $(u,v)$ between
a high out-degree $u$ and a high in-degree node $v$.
We subdivide each such edge with two edges that connect $u$ to $v$ via a helping low degree node $\ell$, i.e., removing the directed edge $(u,v)$ and adding the edges $(u,\ell)$ and $(v,\ell)$. Note that there are at most $m$ such edges,
 so we can distribute the help between low degree nodes in such a way
 that each low degree node helps at most $\avgd$ such edges. Call the resulting graph $G'$.

Accordingly, we also create a new matrix $\cD'$ which, initially, is identical to $\cD$. Then for each $(u,v)$ and $\ell$ as above
we set $\cD'(u,v)=0$, $\cD'(u,\ell) = \cD(u,l)+ \cD(u,v)$ and $\cD'(\ell,v) = \cD(l,v)+ \cD(u,v)$.
Note that $\cD'$ is not a distribution matrix anymore, as the sum of all the entries is more than one,  but it has the following property:
For each high degree node $i$, we have $H(\vec{\cD'}[i]) \le H(\vec{\cD}[i])$ and $H(\cev{\cD'}[i]) \le H(\cev{\cD}[i])$ (see Eq. \eqref{decompos1}).


In the second phase, we construct $\dang$ from $G'$. Consider each node $i$ with high out-degree and create a nearly optimal binary tree $\vec{B}^i$ according to $\vec{\cD'}[i]$ using the method of \cite{Mehlhorn75}. Add an edge from $i$ to the root of $\vec{B}^i$ and delete all the out-edges from $i$ from $G'$. Similarly consider each node $j$ with high in-degree and create a nearly optimal binary tree $\cev{B}_j$ according to ${\cD'}[i]$ using the method of \cite{Mehlhorn75}. Add an edge from $j$ to the root of $\cev{B}_j$ and delete all the in-edges of $j$ from $G'$. This completes the construction of $\dang$.

We first bound the maximum degree in $\dang$.
First consider a low degree node $\ell$, helping an edge $(u,v)$, i.e.,
$u$ is high out-degree and $v$ is high-indegree.
So $\ell$ is part of both $u$'s and $v$'s binary tree, 
hence $\ell$'s
degree increases by at most 6 (two times degree 3 for being an internal node).
Note that $\ell$ needs to help at most $\avgd$ edges itself. 
For each of these $\avgd$ edges, $\ell$'s degree will be 
at most $6$, resulting in a degree of $6\avgd$.
Since $\ell$'s degree was at most $2\avgd$, 
in the worst case, $\ell$ was associate with
$2\avgd$ high in-degree or out-degree nodes, i.e.,
$\ell$ will be present in all these $2\avgd$ trees,
which results in another $6\avgd$ degrees for $\ell$.
In total, $\ell$'s degree is $12\avgd$.
If a node $h$ has both high out-degree and high in-degree, 
then its degree will be two: $h$ is connected to the root
of the tree created of its out-edges and to the root
of the tree created of its in-edges.
If a node $u$ is only a high out-degree node, its degree in $\dang$ is bounded by $6\avgd+1$ (and similarly for a node $u$ which is only  a high in-degree node).
If a node is neither high nor low degree, then its degree in $\dang$ is bounded by $12\avgd$ (originally it was up to $4\avgd$ in $G'$).
We now bound $\epl(\cD,\dang)$.
Recall that from Lemma \ref{epl_tree} and Eq. \eqref{decompos1},  we have,
$$
\epl(\vec{\cD'}[i], \vec{B}_i) \le O(H(Y \mid X=i))
$$
and
$$
\epl(\cev{\cD'}[j], \cev{B}_j) \le O(H(X \mid Y=j))
$$

For each request $(i,j)$ in $\cD$ there are two possibilities for the route on $\dang$: either the edge $(i,j) \in \dang$ is a direct
route, or the route goes via $\vec{B}_i$ or $\cev{B}_j$ or both.
Let $\mathcal{O}$ and $\mathcal{I}$ be the set of high out-degree and in-degree nodes, respectively.
Then:
\begin{align*}
\epl(\cD,\dang) &= \sum_{(u,v) \in \cD} p(u,v) d_{\dang}(u,v) \\
&\le \sum_{ (i,j) \notin \mathcal{O} \cup \mathcal{I}} p(u,v)  \\ & + \sum_{i \in \mathcal{O}} p(i) \epl(\vec{\cD}[i]), \vec{B}_i) \\ & +\sum_{j \in \mathcal{I}} q(j) \epl(\cev{\cD}[j]) \cev{B}_j)\\
&= \sum_{i\notin \mathcal{O}}p(i) + \sum_{j \notin \mathcal{I}} q(j) \\ & +\sum_{i \in \mathcal{O}} p(i) \epl(\vec{\cD}[i]), \vec{B}_i) \\ &  +\sum_{j \in \mathcal{I}} q(j) \epl(\cev{\cD}[j]) \cev{B}_j)\\
&\leq O(H(X \mid Y)+H(Y \mid X))
\end{align*}

This matches our lower bound in Theorem~\ref{th:lowerboundBND}.
\end{proof}

\section{Regular and Uniform Distributions}\label{sec:reguni}

Another large family of distributions for which demand-aware networks can
be designed are regular and uniform distributions $\cD$.
While it is easy to see that
both conditions can be relaxed such that the supported distributions can be
``\emph{nearly} regular'' and
``\emph{nearly} uniform'',
for ease of presentation,
we keep the conditions strict in what follows.

We first establish an interesting connection to spanners.
As we will see, this connection will provide a simple
and powerful technique to design a wide range of
demand-aware networks meeting the conditional entropy lower bound.

\begin{theorem}
\label{th:spanner2design}
Let $\cD$ be an arbitrary (possibly dense) regular and uniform
request distribution.
It holds that if we can find a constant and sparse (i.e., constant distortion, linear sized) spanner for $G_{\cD}$, we can design
a constant degree \dan~$\dang$ providing an expected
path length of
\begin{align}
\epl(\cD,\dang) \le O(H(Y \mid X) + H(X \mid Y))
\end{align}
This is asymptotically optimal.
\end{theorem}

In other words, for regular and uniform distributions, the network
design problem boils down to finding a constant\footnote{To be precise,
a spanner with constant \emph{average} distortion will be sufficient (see Appendix for details). However,
for simplicity, we leave it as a constant
spanner.} sparse spanner:
as we will see, we can automatically transform this spanner
into an efficient network (which may contain auxiliary edges).
The remainder of this section is devoted to the proof of
the theorem.

Assume that $\cD$ is $r$-regular and uniform. Recall that in this case $H(Y \mid X) =  H(X \mid Y) = \log r$, so $\bnd(\cD,\Delta) \geq \Omega(H(Y \mid X))$ where $\Delta$ is a constant.
We now describe how to transform a constant sparse (but potentially irregular)
spanner for $G_{\cD}$ into a constant-degree host network $\dang$
with $\epl(\cD,\dang) \leq O(\log r)$. This will be done using a similar
degree reduction technique as discussed earlier (in the proof of Theorem~\ref{lem:sparsetransform}).

%


\begin{lemma}\label{lemma:helper}
Let $G$ be a graph of maximum degree $\maxd$
and an average degree $\avgd$.
Then, we can construct a graph
$G'$ with maximum degree $8\avgd$ which is a graph metric $\log \maxd$-spanner of $G$,
 i.e.,  $d_{G'}(u,v)\leq 2 \log \maxd \cdot d_{G}(u,v)$.
\end{lemma}
\begin{proof}
Let us call the $n/2$ nodes with the lowest degree in $G$
the  \emph{low degree} nodes and the remaining nodes
\emph{high degree} nodes.
By the pigeon hole principle, each low degree node has a
degree at most $2\avgd$.
The construction of $G'$ proceeds in two phases.
In the first phase we take every edge between high degree nodes
$u,v$ and subdivide it with two edges that connect $u$ to $v$ via a
helping low degree node $\ell$, i.e., removing the edge $(u,v)$ and adding the edges
$(u,\ell)$ and $(v,\ell)$. Note that there are at most $m$ edges connecting high degree nodes so we can distribute the help between
low degree nodes such that each
low degree node helps to at most $\avgd$ such edges.

In the second phase we consider each high
degree node $u$ and replace the set of
edges between $u$ and its neighbors, $\Gamma(u)$,
with a balanced binary tree that connects $u$ as the root and
$\Gamma(u)$ as
remaining nodes of the tree. Denote as $B_u$ this tree and note that
the height of $B_u$ is at most $\log( \card{\Gamma(u)}+1)$.
We leave edges between low degree nodes as in $G$.

Let us analyze the degrees in $G'$. Since every high degree
node $u$ in $G'$ only connects to low degree nodes, it is only a
member of  $B_u$ and
its degree reduces to 2 in $G'$. Now consider a low degree node $\ell$:
for each edge $(u,v)$ it helps, $\ell$ participates in $B_u$ and $B_v$.
Hence, its degree increases by at most 6 in $G'$ compared to $G$. Overall, for helping high degree nodes, the degree of $\ell$ can increase by $6\avgd$.
Together with its original neighbors from $G$,
the degree of $\ell$ in $G'$ can be at most $8\avgd$.

Next consider the distortion of $G'$. The distortion between neighboring low degree nodes is one. The distortion between neighboring high degree nodes
is at most twice $\log \maxd$ and the distortion between a neighboring high and low degree is at most $\log \maxd$.

So, $d_{G'}(u,v)\leq 2 \log \maxd \cdot d_{G}(u,v)$ for all $u$, $v$ in $G'$.
\end{proof}

We will apply Lemma~\ref{lemma:helper} to prove Theorem \ref{th:spanner2design}.
\begin{proof}[Proof of Theorem \ref{th:spanner2design}]
Let $\spanner$ be a constant and sparse spanner of $G_{\cD}$ ($G$ could be either a subgraph or a metric spanner of max degree
asymptotically not larger than $G_{\cD}$) of degree at most $r$.
 Lemma~\ref{lemma:helper} then tells us how to transform $\spanner$ to
 a $\dan$  $\dang$ of degree $\avgd$.
Since $\spanner$ is a constant spanner there is a constant $c$ such that,
\begin{align}\label{eq:onaverage}
\epl(\cD,\spanner) = \sum_{(u,v) \in \cD} \p(u,v)\cdot d_{\spanner}(v,v) = c
\end{align}
Since $\spanner$ has maximum degree $r$,
according to Lemma~\ref{lemma:helper}, it has a graph metric spanner $\dang$ such that,
the distance of any source-destination pair of $G(\cD)$ in $\dang$ is at most $2 \log r$
times their distance in $\spanner$. Hence:
\begin{align*}
\epl(\cD,\dang) &= \sum_{(u,v) \in \cD} \p(u,v)\cdot d_{\dang}(u,v) \\ & \leq \sum_{(u,v) \in \cD} \p(u,v)\cdot d_{\spanner}(u,v)\cdot 2 \log r \\
& \leq \log r \cdot ~\epl(\cD,\spanner) \\ & = O(\log r) = O(H(Y \mid X))
\end{align*}
The last equality holds since $\cD$ is $r$-regular and uniform.
The bound is asymptotically optimal when $\Delta$ is a constant: it matches our lower bound in Theorem~\ref{th:lowerboundBND}.
\end{proof}

Theorem~\ref{th:spanner2design} allows us to simplify the design of asymptotically optimal networks for
uniform and regular distributions $\cD$ where $G_{\cD}$ has a constant sparse spanner.
In particular, the approach can be used to design optimal
networks for the following large families of distributions which are known to have a constant and sparse graph spanners.
\begin{corollary}\label{cor:uniformreg}
Let $\cD$ describe a uniform and regular
communication request distribution.
Then, it is possible to generate a constant degree  \dan~$\dang$
such that
\begin{align}
\epl(\cD,\dang) \le O(H(Y \mid X) + H(X \mid Y))
\end{align}
in the following scenarios:
\begin{itemize}
\item If, for a constant $c \ge 1$, $G_{\cD}$ has a minimum degree $\delta\geq n^{\frac{1}{c}}$.\footnote{In this case the constant in the $O$ notations depends linearly on $c$.}
\item If $G_{\cD}$ forms a hypercube with $n \log n$ edges.
\item If $G_{\cD}$ forms a (possibly dense) chordal graph.
\end{itemize}
\end{corollary}

\noindent See Appendix for proof and details.

We round off our study of uniform and regular distributions
by considering one more interesting family of
(possibly very dense) distributions:
distributions $\cD$ which describe a bounded and \emph{local} doubling dimension,
note that this family is more general than the standard bounded doubling dimension graphs which are sparse.


First, recall that a metric space $(V, d)$ has a constant doubling dimension
if and only if there exists a constant $\lambda$ such that
every ball of radius $r$ in $V$ can be covered by $\lambda$ balls of half the radius
$r/2$, for all $r\geq 1$. In general, the smallest $\lambda$ which satisfies this property for a metric space is called \emph{doubling constant} and $\log_2 \lambda$ is called the \emph{doubling dimension} \cite{Chan16,doubling-zvi,gupta2003bounded,har2006fast}.
A metric space is called \emph{bounded} (a.k.a.~constant or low) doubling dimension if $\lambda$ is a constant.
There has been a wide range of work on spanners for bounded doubling dimension metrics \cite{Chan16,chan2009small,gupta2003bounded,har2006fast}.
However, if the metric is imposed by a graph metric (via shortest paths) then a bounded doubling dimension implies that the graph is nearly regular, of bounded (constant) degree and sparse. Theorem \ref{lem:sparsetransform} already solved the case of sparse $G_{\cD}$, even for non-uniform and irregular distributions.

In the following, however, we are interested in a
more general notion of doubling dimension,
which allows a higher density, unbounded degree:
we call it \emph{locally-bounded doubling dimension}:

\begin{definition}]\label{def:locallydd}
$G_{\cD}$ implied by the distribution $\cD$ has a \emph{locally-bounded doubling dimension}
if and only if there exists a constant $\lambda$ such that the 2-hop neighbors of any node $u$ are covered by
at most $\lambda$  1-hop neighbors. Formally, for each $u\in V$, there exists a set of nodes ${y_1, y_2,...y_\lambda}$, such that:
$$B(u,2)\subseteq\bigcup_{i=1}^{\lambda} B(y_i,1)$$
where $B(u,r)$ are the set of nodes that are at distance at most $r$-hops from $u$ in $G_{\cD}$.
\end{definition}

Clearly, every bounded doubling dimension graph is also of locally-bounded doubling dimension, but the converse is not true.
In particular, the latter enables graphs which could be dense, with unbound degree, and possibly with irregularity of degree.




In the remainder of this section, we will prove the following theorem.
\begin{theorem}\label{thm:ldd}
Let $\cD$ describe a uniform and regular
communication request distribution of locally-bounded doubling dimension.
Then it is possible to design a constant degree  \dan~$\dang$
such that
\begin{align}
\epl(\cD,\dang) \le O(H(Y \mid X) + H(X \mid Y))
\end{align}
This is asymptotically optimal.
\end{theorem}
\begin{proof}
Again, our proof strategy is to employ
Theorem~\ref{th:spanner2design}. Accordingly,
we show that a constant sparse spanner exists for locally-bounded
doubling dimension networks.
In particular, we will design this spanner based on an $\epsilon$-net construction.
We first recall the definition of \emph{$\epsilon$-nets}~\cite{Chan16}. 
\begin{definition}[$\epsilon$-net] \label{def:r-net}
A subset $V'$ of $V$ is an~\emph{$\epsilon$-net} for a graph $G=(V,E)$ if it satisfies
the following two conditions:
\begin{enumerate}
  \item for every $u$, $v \in V'$, $d_G(u, v)> \epsilon$
  \item for each $w \in V$, there exists at least one $u \in V'$
  such that, $d_G(u, w)\leq \epsilon$
\end{enumerate}
\end{definition}

Let $G_{\cD}=(V,E)$ be a locally-bounded doubling dimension network.
We now first construct a spanner $\spanner'$ of $G_{\cD}$
which is a subgraph of $G_{\cD}$,
using the following $(\epsilon=2)$-net:
we sort all nodes according to decreasing
(remaining) degrees, and iteratively select the high-degree nodes into the 2-net
 one-by-one and remove their 2-neighborhoods.
Clearly, after this process, each node is either part of the 2-net
or has a 2-net node at distance at most 2-hops,
and we have computed a legal 2-net.

To form the spanner $\spanner$, we next
arbitrarily match each node $u$ not in the 2-net to one of its nearest
2-net nodes
$v$,
and select the edges along a shortest path from $u$ to $v$ into the spanner $\spanner$.
This results in a forest of connected components (2-layered stars).
We call these connected components \emph{clusters}
and the corresponding nodes in the 2-net \emph{cluster heads}.
We denote the cluster associated to the net node $u$ by $Cl(u)$.

We next connect the connected clusters to each other, in a sparse manner.
Towards this end, we connect each pair of clusters, with an arbitrary single edge,
if they contain at least one pair of communicating nodes in $G_{\cD}$.
We can claim the following.
\begin{lemma}\label{th:spanner}
$\spanner$ is a constant and sparse spanner of $G_{\cD}$ (with distortion $9$) .
\end{lemma}
\begin{proof}
Let $(u,v)$ be an edge in $G_{\cD}$ and $u\in Cl(u)$, $v\in Cl(v)$.
By construction, there are nodes $x \in Cl(u)$ and $y \in Cl(v)$
that are connected by an edge in $\spanner$, and hence there is a path $u, C(u), x, y, C(v), v$ in  $\spanner$.
Therefore, $d_S(u,v) \le d_S(u,Cl(u)) + d_S(Cl(u),x) + d_S(x,y)+ d_S(y,Cl(v)) +d_S(Cl(v),v)  \le 9 $.

The spanner is also sparse:
in a nutshell, due to the 2-net properties, we know that the
distance between communicating cluster heads is at most 5:
since in a locally doubling dimension graph the number of cluster heads
at distance 5 is bounded, only a small number of
neighboring clusters will communicate.
More formally, after associating each
node to some unique cluster,
we have a linear number of edges in the spanner.
Next we bound the number of outgoing edges from each
cluster.
Let $(u,v)$ be such an edge where $u \in Cl(u)$, $v\in Cl(v)$.
Let the cluster heads of $Cl(u)$ and $Cl(v)$ be $i$ and $j$, respectively.
By construction $i$ and $j$ are at most at distance 5 in $G_{\cD}$, i.e., $d_{G_{\cD}} (i,j) \leq 5$.
So, if we can bound the number of 2-net nodes which lie
within 5 hops from some net node $i$,
it will give us a bound on the number of edges which we add between $Cl(u)$ and other clusters.
According to Definition~\ref{def:locallydd}, all the two hop neighbors of $i$ can be covered within
one hop neighbors of $\lambda$ nodes, where
$\lambda$ is the corresponding doubling constant. If we consider two hop neighbors of all
these $\lambda$ many nodes, they cover all the 3 hop neighbors of $i$. To cover the
2 hop neighbors of each of
these nodes, we again require one hop neighbors of $\lambda$ nodes.
So, to cover all 3 hop neighbors of $i$, we require at most $\lambda^2$
one hop neighbors. Inductively, by repeating this argument, we require one hop
neighbors of at most $\lambda^4$ nodes to cover all the 5 hop neighbors
of $i$. Since we constructed a 2-net, each of these $\lambda^4$ balls
with radius one contains at most one 2-net node. Hence there are at most $\lambda^4$
2-net nodes which are at a distance 5 hops or less from $i$. Consequently, there are
at most $\lambda^4$ inter-cluster edges associated to some cluster $Cl(u)$, or cluster head $i$.
Since there can not be more than linear number of clusters, hence the number of
edges in $\spanner'$  is also linear.
\end{proof}

Using Lemma \ref{th:spanner} we can directly use Theorem \ref{th:spanner2design} and conclude the proof of Theorem~\ref{thm:ldd}.
\end{proof}

In fact, it turns out that if we consider a \emph{metric} spanner, and by using auxiliary edges,
we can improve the above distortion and constract better constant sparse spanner $\spanner'$.
The idea is to add inter-cluster edges only between the cluster heads.
This will reduce the distortion to $5$ while keeping the same number to total edges.
The degree of each node in $\spanner'$ will increase by at most a constant, $\lambda^4$.
Adjusting Theorem \ref{th:spanner2design} respectively to support metric spanners (and only a subgraph spanner) will enable
us to use  $\spanner'$ instead of $\spanner$.

\section{Conclusion}\label{sec:conclusion}

This paper initiated the study of a fundamental
network design problem. While our work is motivated in particular
by emerging technologies for more flexible datacenter interconnects
as well as by peer-to-peer overlays,
we believe that our model is very natural and of interest beyond
this specific application domain considered in this paper.
For example, the design of a sparse, bounded-degree and distance-preserving network
can also be understood from the perspective of graph sparsification~\cite{spielman2004nearly}: the designed
network can be seen as a compact representation of the original network.

The subject of
bounded network design offers several interesting avenues
for future research. In particular, while we presented asymptotically
optimal network design algorithms for a wide range of distributions and
while we believe that the entropy is the right measure to assess network designs,
there remain several (dense) distributions for which the quest for
optimal network designs remains open, perhaps also requiring
us to explore alternative flavors of graph entropy.

\noindent  \textbf{Acknowledgments.}
We would like to thank Michael Elkin for many inputs and discussions.

{

}


\appendix

\section*{Appendix}

\section{Deferred Proofs}

\textbf{Lemma \ref{lem:DeltaregularH}}
\begin{proof} [Proof of Lemma \ref{lem:DeltaregularH}]
Let $T$ be any $\Delta$-ary tree over the probability $\bar{p}$ (with nodes $1,2, \dots n$).  To each node which has less than $\Delta-1$ children in $T$,
add leaves to make the number of its children $\Delta$. Call this tree~$T'$.
There would be~$n(\Delta-1)+1$ leaves in $T'$. This can easily be shown by induction on the number of internal nodes.
The frequency to access the internal nodes of $T'$ remains $\bar{p}$. The frequency to access the leaves will be $\bar{p'}= \bar{0}$, namely
$p'_j=0$ for all leafs $j=1,2,...,n(\Delta-1)+1$.
Let~$b_i$,~$a_j$ be the distances of the internal nodes and the leaves respectively, from the root.
The expected path length to reach nodes in $T'$ from the 
root would be~$\sum p_i b_i +\sum p'_ja_j = \sum p_i b_i$ which is $\epl(\bar{p},T)$.
We now define:
$$
L=\sum_{i=1}^{n}(\Delta+1)^{-(b_i+1)}+\sum_{j=1}^{n(\Delta-1)+1}(\Delta+1)^{-a_j}
$$
Using induction, it can be easily shown that~$L=1$ and hence~$\log L=0$. Define,
$$
f'_i=(\Delta+1)^{-(b_i+1)} ~~~\text{for} ~~~1\leq i\leq n
$$
and
$$
f''_j=(\Delta+1)^{-a_j} ~~~\text{for} ~~~1\leq j\leq n(\Delta-1)+1
$$
So,~$\sum_{i=1}^n f'_i + \sum_{j=1}^{n(\Delta-1)+1}f''_j = 1$ and  consequently~$\{f'_1,...f'_n,f''_1,...f''_{n(\Delta-1)+1}\}$ is a distribution.
Recall Gibbs' inequality~\cite{cover2012elements} which states that $\sum p_i\log{1/p_i} \le  \sum p_i\log{1/f_i}$ for any distribution $\bar{p}$ and $\bar{f}$.
Thus, we can have:
\begin{align*}
H_\Delta(\bar{p})&=\sum p_i\log{1/p_i}\\
&=\sum p_i\log{1/p_i} + \sum p'_j\log{1/p'_j}\\
&\leq \sum p_i\log{1/f'_i} + \sum p'_j\log{1/f''_j}\\
&=\sum p_i\log(\Delta+1)^{(b_i+1)}+\sum p'_j\log(\Delta+1)^{a_j}\\
&=\log(\Delta+1)(\sum p_i(b_i+1)+\sum p'_ja_j\\
&=\log(\Delta+1)(\epl(\bar{p},T) + 1)
\end{align*}
\end{proof}

\subsection*{Corollary \ref{cor:uniformreg}}

\begin{proof}[Proof of Corollary \ref{cor:uniformreg}]
We prove the claims in turn.\\

\noindent \textbf{Case $G_{\cD}$ has a minimum degree $\Delta \geq n^{\frac{1}{c}}$:}
For this distribution $\cD$, we have, $H_{\Delta}(Y\mid X)=H_{\Delta}(X\mid Y) \ge \frac{1}{c}\log_{\Delta}n$.
Create a $\Delta$-ary tree $\dang$ with the nodes of $G_{\cD}$. Then,
\begin{align*}
\epl(\cD,\dang) &= \sum_{(u,v) \in \cD} \p(u,v)\cdot d_{\dang}(u,v)\\
&\leq \sum_{(u,v) \in \cD} \p(u,v)\cdot 2 \log_{\Delta}n\\
&\leq 2 \log_{\Delta}n\\
&\leq 2c\cdot H_{\Delta}(Y\mid X)
\end{align*}

\noindent \textbf{Case hypercube.}
Follows directly from Theorem~\ref{th:spanner2design}
and the fact that hypercubes admit a sparse 3-spanner~\cite{Peleg89},
allowing us to design a~$O(\log\log n)$ (metric) spanner of bounded degree.\\

\noindent \textbf{Case chordal graphs.}
Follows from Theorem~\ref{th:spanner2design}
and the fact that chordal graphs have a
constant sparse spanner~\cite{Peleg1989}.
\end{proof}

\section{Notions of Distortion}
In the spanner problem, the goal is to find a sparse subgraph~$\spanner=(V,E')$ of~$G$, i.e.,
$E'\subseteq E$ with~$|E'|\leq O(n)$ which
approximately preserves the distances of~$G$ despite having less edges.
Usually, the following notion of average distortion~\cite{Chan2006}  is considered
and referred to as the \emph{all-pairs distortion}:

\begin{definition}[All-Pairs Distortion ($\apd$)]\label{def:distortion-all}
The average all-pairs distortion on a spanner~$\spanner$ of a graph~$G$ is
$$\apd(G,\spanner) =\frac{1}{\binom{n}{2}} \sum_{\{u,v\}\in \binom{V}{2}}
\frac{d_{\spanner}(u,v)}{d_G(u,v)}$$
\end{definition}

We in this paper are only interested in preserving distances
between \emph{communicating neighbors} in~$G$,
henceforth defined as the \emph{neighborhood distortion}:

\begin{definition}[Neighborhood Distortion ($\nd$)]\label{def:distortion-neigh}
The average neighborhood distortion on a spanner~$\spanner$ of a graph~$G$ (with~$m$ edges) is,
\begin{eqnarray*}
\nd(G,\spanner) & =\frac{1}{\card{E(G)}} \sum_{\{u,v\}\in E(G)} \frac{d_{\spanner}(u,v)}{d_G(u,v)} \\ &= \frac{1}{m} \sum_{\{u,v\}\in E(G)} d_{\spanner}(u,v)
\end{eqnarray*}
\end{definition}

Next we claim that these two notions of distortion are indeed different, 
that is, low all-pairs distortion
does not imply a low neighborhood distortion; and vice versa.

\begin{claim}
There is a family of graphs $G_n$ and a corresponding family of spanners $S_n$ (where $n$ is the size of the graph and $S_n$ is a spanner of $G_n$) where
\begin{align}
\lim_{n \rightarrow \infty} \frac{\apd(G_n, S_n)}{\nd(G_n, S_n)} = \infty
\end{align}
\end{claim}

\begin{claim}
There is a family of graphs $G_n$ and a corresponding family of spanners $S_n$ (where $n$ is the size of the graph and $S_n$ is a spanner of $G_n$) where
\begin{align}
\lim_{n \rightarrow \infty} \frac{\nd(G_n, S_n)}{\apd(G_n, S_n)} = \infty
\end{align}
\end{claim}

We will show this by examples.
First consider Figure \ref{f:fig5} (a).
There is a~$\Theta(\sqrt{n})$-sized clique
in the center, and each of those clique nodes is associated with a
line containing~$\Theta(\sqrt{n})$ nodes.
To compute the optimal tree spanner with maximum degree~$\Delta$, we
turn the clique nodes into a~$\Delta$-regular tree
of diameter~$\Theta(\log_\Delta{\sqrt{n}})=O(\log_\Delta{n})$.
The nodes remain connected with the corresponding lines.

\begin{figure*}
    \begin{center}
	\begin{tabular}{cc}
	   \includegraphics[width=.45\textwidth]{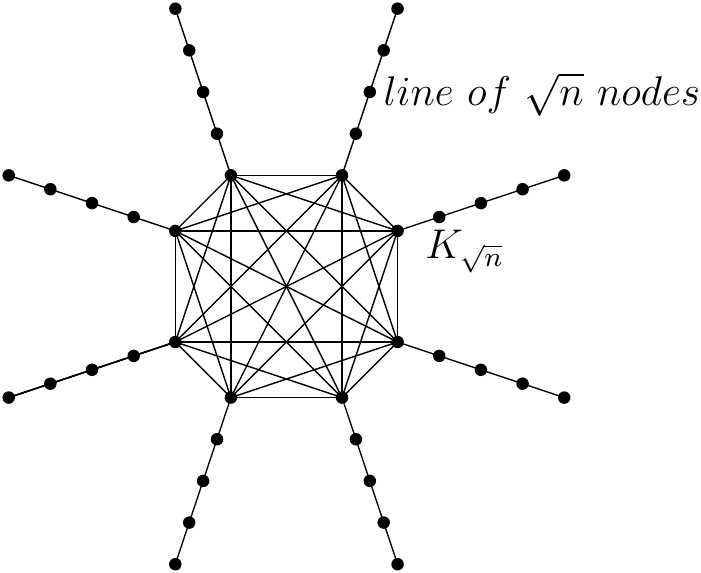} &
	    \includegraphics[width=.45\textwidth]{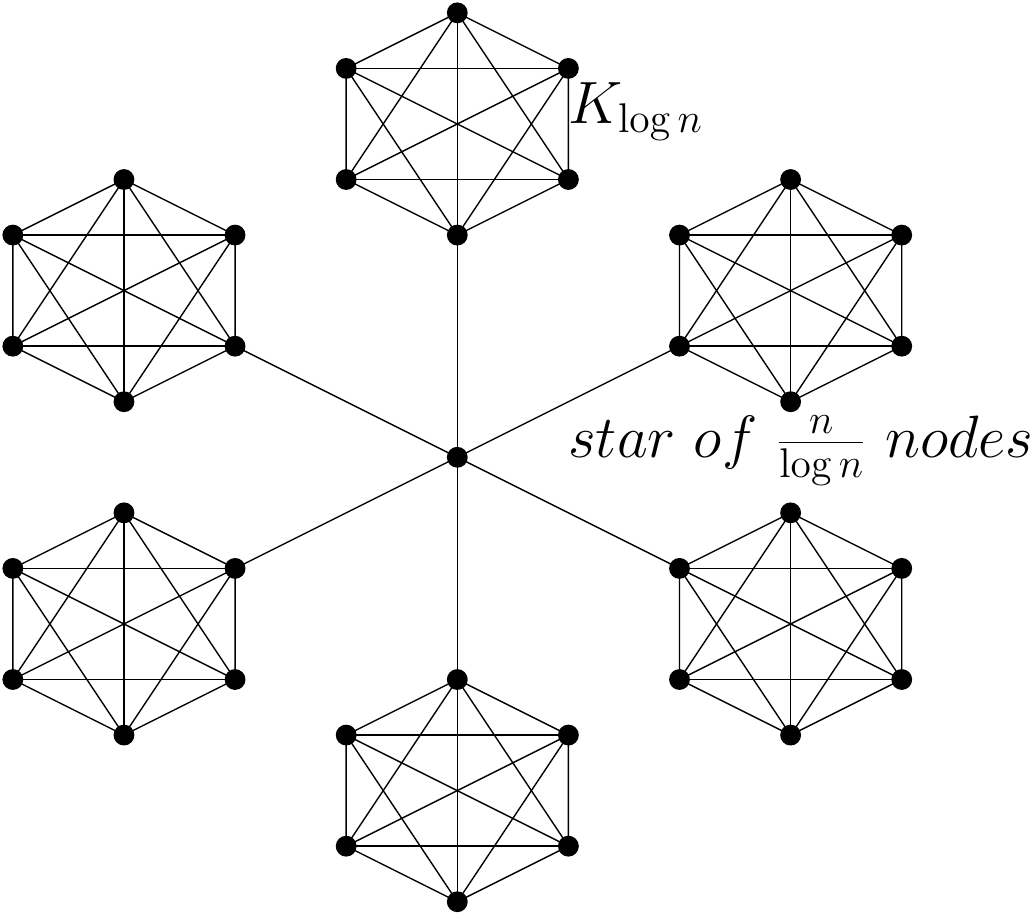} \\
	    (a) & (b) 
	\end{tabular}
	   \end{center}
        \caption{(a) Different distortions on tree spanner w.r.t.~different definitions.
        (b) Different distortion on tree network design (with auxiliary edges) w.r.t.~different definitions}
\label{f:fig5}
\end{figure*}

The asymptotic distortion w.r.t.~Definition~\ref{def:distortion-neigh} is:
$$
\frac{n \cdot\log_\Delta{n}+n \cdot 1}{n}=\Theta(\log_\Delta{n})
$$
Now we discuss about all pair distortion on the same spanner for this graph.
Consider any two nodes which belong to different lines, but are also a member of the clique.
Their distance in the spanner may become~$\log \sqrt n$. So, according to Definition~\ref{def:distortion-all},
$d_{\spanner}(u,v)/d_G(u,v)$ is equal to~$\frac{1}{2}\log n$. Now we provide an upper bound on number of such
pairs~$\varphi$  whose distance can be up to~$O(\log n)$ times their earlier distance. Consider all the nodes on
all the lines which are within distance~$\log n$ from the corresponding clique node.
On the original graph, distances between any two such nodes were in the range~$[1,2\log n +1]$. The number of such node pairs is~$\sqrt n \log^2 n$. Clearly,~$\varphi<\sqrt n\log^2 n$.
Now consider any node on a line which is at least at a distance~$(1+\log n)$ from the corresponding
clique node on the line. The distance from this node to any other node on any other line was at least~$(2+\log n)$.
On the spanner, this distance can be at most~$1+2\log n$. So, for all such node pairs,
$d_{\spanner}(u,v)/d_G(u,v)<2$. Hence, according to Definition \ref{def:distortion-all}, all pair distortion
becomes constant as given in the following expression.

$$
\frac{\sqrt n \log^2 n\cdot \log n + (n^2-\sqrt n \log^2 n)\cdot 2}{n^2}=\Theta(1)
$$

Now look at Figure~\ref{f:fig5} (b). There is a star of size~$n/\log{n}$
in the center, and each of the~$n/\log{n}$ nodes is associated with a clique of
size~$\log{n}$. Thus, in total, there are~$n\log{n}$ edges.
To compute a tree spanner of degree~$\Delta=\log{n}$,
we replace the cliques consisting of~$\log{n}$ nodes with stars of
size~$\log{n}$ nodes; the star of~$n/\log{n}$ nodes in the center is
transformed into a~$\Delta$-regular tree whose diameter is~$\Theta(\log{n}/\log{\log{n}})$.
Then each tree node is associated with the root of the star
corresponding to its associated clique. This tree spanner contains auxiliary edges too.
Then, the asymptotic distortion w.r.t.~Definition~\ref{def:distortion-neigh}
is:
$$
\frac{\frac{n}{\log n}\log^2n+\frac{n}{\log n}\cdot\frac{\log n}{\log \log n}}{n\log n}=O(1)
$$
In contrast, the distortion w.r.t.~Definition~\ref{def:distortion-all} is
$\Omega(\log n/\log \log n)$ since all pairs from the two different cliques now suffer
a distortion of~$\Theta(\log n/\log \log n)$, and there are~$O(n^2)$ such pairs.
\section{Section 5}
\begin{corollary}\label{cor:onaverage}
Theorem \ref{th:spanner2design} holds even if there exists a sparse spanner~$\spanner$ with constant neighborhood distortion instead of having a constant spanner.
\end{corollary}
\begin{proof}
If the request distribution~$\cD$ is uniform, i.e.,~$p(i,j)=1/m$ for all the~$m$ non-zero entries of the matrix~$M_{\cD}~$, then from Definition~\ref{def:distortion-neigh} and from our objective function,
$$
\epl(\cD, \spanner)  = \nd(G, \spanner)
$$
Hence~$\nd(G,\spanner)$ is constant, which implies that~$\epl(\cD, \spanner)$ is also constant i.e.,
Equation~\ref{eq:onaverage} holds if~$\nd(G,\spanner)$ is constant. 
\end{proof}

%
%

\end{document}